%% file: final.tex
\documentclass[copyright,creativecommons]{eptcs}
\providecommand{\event}{TARK 2019} 
\usepackage{breakurl}             
\usepackage{underscore}           
\usepackage{amsfonts}
\usepackage{amsmath}
\usepackage{xspace}
\usepackage{comment}
\usepackage{changebar}
\usepackage{amssymb,amsthm}
\usepackage{tikz}
\usetikzlibrary{arrows,backgrounds,automata,shapes,patterns}

\input{macros.tex}

\title{Social Choice Methods for Database
  Aggregation}
\author{Francesco Belardinelli
\institute{Department of Computing, Imperial College London, U.K.}
\institute{Laboratoire IBISC, Universit\'e d'Evry, France}
\email{francesco.belardinelli@imperial.ac.uk}
\and
 Umberto Grandi
\institute{IRIT, University of Toulouse\\France}
\email{umberto.grandi@irit.fr}
}
\def\titlerunning{Social Choice Methods for Database
  Aggregation}
\def\authorrunning{F. Belardinelli \& U. Grandi}
\begin{document}

\maketitle



\begin{abstract}
Knowledge can be represented compactly in multiple ways, from a set of
propositional formulas, to a Kripke model, to a database. In this
paper we study the aggregation of information coming from multiple
sources, each source submitting a database modelled as a first-order
relational structure. In the presence of integrity constraints, we
identify classes of aggregators that respect them in the aggregated
database, provided these are satisfied in all individual databases. We
also characterise languages for first-order queries on which the
answer to a query on the aggregated database coincides with the
aggregation of the answers to the query obtained on each individual
database.  This contribution is meant to be a first step on the
application of techniques from social choice theory to knowledge
representation in databases.
\end{abstract}

\section{Introduction}

\input{intro.tex}

\section{Preliminaries on Databases}\label{sec:preliminaries}

\input{preliminaries.tex}

\section{Aggregators} \label{aggregators}

\input{aggregators.tex}

\section{The Axiomatic Method}\label{sec:axioms}

\input{axioms.tex}





\section{Collective Rationality}\label{sec:collectiverationality}

\input{collectiverat.tex}

\section{Aggregation and Query Answering}\label{sec:queries}

\input{query.tex}





\section{Conclusions and Related Work}\label{sec:conclusions}

\input{conc.tex}

\medskip\noindent\textbf{Acknowledgements.}  F.~Belardinelli
acknowledges the support of the ANR JCJC Project SVeDaS
(ANR-16-CE40-0021), and U.~Grandi the support  of the ANR JCJC project SCONE (ANR 18-CE23-0009-01).

\bibliographystyle{eptcs}
\bibliography{databases}
\end{document}

%% file: macros.tex
\newcommand{\adom}{\textit{adom}}

\renewcommand{\phi}{\varphi}

\newtheorem{definition}{Definition}
\newtheorem{theorem}{Theorem}
\newtheorem{lemma}[theorem]{Lemma}
\newtheorem{proposition}[theorem]{Proposition}

\newtheorem{question}[theorem]{Question}

\newcommand{\U}{\mathcal U}
\newcommand{\D}{\mathcal D}
\newcommand{\N}{\mathcal N}
\renewcommand{\L}{\mathcal L}
\newcommand{\F}{\mathcal F}

\newcommand{\myi}{{(i)}\xspace}
\newcommand{\myii}{{(ii)}\xspace}

\newcommand{\lit}{\text{\it lit}}

\newcommand{\argmin}{\operatornamewithlimits{argmin}}

\newtheorem{exampleAux}{Example} 
\newenvironment{example}
				{\begin{exampleAux}\normalfont}
				{\phantom{x} \hfill
					$\blacksquare$\end{exampleAux}}

%% file: intro.tex
Aggregating information coming from multiple sources is a
long-standing problem in both knowledge representation and
multi-agent systems (see, e.g., \cite{vanHarmelen2007}).  Depending on
the chosen representation for the incoming pieces of knowledge or
information, a number of competing approaches has seen the light in
these literatures.  Belief
merging \cite{LiberatoreSchaerf1998,KoniecznyPinoPerezJLC2002,KoniecznyLangMarquisAIJ2004}
studies the problem of aggregating propositional formulas coming from
different agents into a set of models, subject to integrity
constraints.  Judgment and binary
aggregation \cite{EndrissHBCOMSOC2016,DokowHolzmanJET2010,GrandiEndrissAIJ2013}
asks individual agents to report yes/no opinions on a set of
logically-related binary issues -- the agenda -- in order to take a
collective decision.  Social welfare functions, the cornerstone
problem in social choice theory (see, e.g., \cite{Arrow1963}), can
also be viewed as mechanisms to merge conflicting information, namely
the individual preferences of voters expressed in the form of linear
orders over a set of alternatives.  Other examples include graph
aggregation \cite{EndrissGrandiAIJ2017}, multi-agent
argumentation \cite{BoothEtAlKR2014,CaminadaPigozzi2011,ChenEndrissTARK2017},
ontology merging \cite{PorelloEndrissJLC2014}, and clustering
aggregation \cite{GionisEtAlTKDD2007}.

In this work we take a general perspective and represent individual
knowledge coming from multiple sources as a profile of databases,
modelled as finite relational
structures \cite{AbiteboulHV95,MaierUV84}.  Our motivation lies
inbetween two possibly conflicting views on the problem of information
fusion.  On the one hand, the study of information merging (typically
knowledge or beliefs) in knowledge representation has focused on the
design of rules that guarantee the consistency of the outcome, with
the main driving principles inspired from the literature on belief
revision\footnote{Albeit we acknowledge the work
of \cite{DoyleWellmanAIJ1991, MaynardZhangLehmannJAIR2003}, which
aggregate individual beliefs, modelled as plausibility orders, in an
"Arrovian" fashion.}.  On the other hand, social choice theory has
focused on agent-based properties, such as fairness and
representativity of an aggregation procedure, paying attention as well
on possible strategic behaviour by either the agents involved in the
process or an external influencing source.  While there already have
been several attempts at showing how specific merging or aggregation
frameworks could be simulated or subsumed by one another (see,
e.g., \cite{GrandiEndrissIJCAI2011,DietrichList2007a,GregoireKonieczny2006,EveraereEtAl2015}),
a more general perspective allows us to find a compromise between the
two views described above.

\textbf{Our Contribution.} Our starting point is a set of finite relational structures on the
same signature, coming from a group of agents or sources. Then, our
research problem is how to obtain a collective database summarising
the information received. Virtually all of the settings mentioned
above (beliefs, graphs, preferences, judgments, \ldots) can be
represented as databases, showing the generality of our approach.  We
propose a number of rules for database aggregation, some inspired by
existing ones from the literature on computational social choice and
belief merging, as well as a new one adapted from representations of
incomplete information in databases \cite{Libkin15}.  We privilege
computationally friendly aggregators, for which the time to determine
the collective outcome is polynomial in the
individual input received.

We first evaluate these rules axiomatically, using notions imported
from the literature on social choice, to provide a first
classification of the agent-based properties satisfied by our proposed
rules.  Then, when integrity constraints are present, we study how to
guarantee that a given aggregator ``lifts'' the integrity constraint
from the individual to the collective level, i.e., the aggregated
databases satisfy the same constraints as the individual ones.
Specifically, we investigate which rules lift
classical integrity constraints from database theory, such as
functional dependencies, referential integrity and value
constraints.  Finally, since databases are typically queried using
formulas in first-order logic, a natural question to ask in a
multi-agent setting is whether the aggregation of the individual
answers to a query coincides with the answer to the same query on the
aggregated database.
We provide a partial answer to this important problem, by identifying sufficient conditions on the
first-order query language.

\textbf{Related Work.} 
 While we are not aware of any application of methods from social
choice theory to database aggregation, possibly the closest approach
to ours is the work of Baral \emph{et
al.} \cite{BaralKM91,BaralEtAl1992} and
Konieczny \cite{Konieczny2000}.  In \cite{BaralKM91} the authors
formalize the notion of combining knowledge bases, which are
represented as normal Horn-logic programs.  These investigations were
further pursued in \cite{BaralEtAl1992}, which considers the problem
of merging information represented in the form of first-order
theories, taking a syntactic rather than a semantic approach (as we do
here), and focusing on finding maximally consistent sets of the union
of the individual theories received.  In doing so, however, the
authors privilege the knowledge representation approach, and have no
control on the set of agents supporting a given maximally consistent
set rather than another.  In \cite{Konieczny2000}, the author applies
techniques from belief merging to the equivalent problem of
aggregating knowledge bases of first-order formulas, proposing a
number of rules analysed axiomatically.
%
Both contributions stem from a long tradition on combining
inconsistent theories, especially in the domain of paraconsistent
logics \cite{BlairS89,Subrahmanian92}.  However, all these approaches
focus on merging syntactic representations (e.g., logic programs,
first-order theories), while here we operate on semantical instances,
i.e., databases. We also mention the work of Lin and Mendelzon \cite{AlbertoEtAl1996}, proposing an AGM-style
approach to merge first-order theories under constraints, reminescent of the distance-based rules
that we will consider in Section~\ref{aggregators}.

Mildly related to the present work is the literature on database
repairs. Here the focus is on principles of minimal change, which is
the aspiration to keep the recovered data as faithful as possible to
the original (inconsistent) database \cite{ArieliDB07}. Our perspective is different, as we analyse aggregation rather then repairing. Nonetheless, we also consider distance-based procedure.

More recently, connections between social choice theory and database
querying have been explored in \cite{KimelfeldEtAlIJCAI2018}, which
enriches the tasks currently supported in computational social choice
by means of relational databases, thus allowing for sophisticated
queries about voting rules, candidates, and voters. Here our aim is
symmetric, as we rather apply methods and techniques from
computational social choice to database theory.

An overview of the results presented hereafter can be found
in \cite{BelardinelliG19}, which introduces the question of database
aggregation and defines some aggregation procedures. Here we
extend \cite{BelardinelliG19} by considering in detail the problems
pertaining to collective rationality through lifting of integrity
constraints in Section~\ref{sec:collectiverationality}, as well as
aggregation and query answering in Section~\ref{sec:queries}.

\textbf{Structure of the Paper.}
In Section~\ref{sec:preliminaries} we present basic notions
on databases and integrity constraints. In Sections~\ref{aggregators}
and~\ref{sec:axioms} we introduce several database aggregation
procedures, and we analyse them by proposing a number of axiomatic
properties.  Sections~\ref{sec:collectiverationality} and
~\ref{sec:queries} contains our main results on the lifting of
integrity constraints and aggregated query answering.
Section~\ref{sec:conclusions} concludes the paper.

%% file: preliminaries.tex

In this section we introduce basic notions on databases that we will
use in the rest of the paper. In particular, we adopt a relational
perspective~\cite{AbiteboulHV95} and present databases as finite
relational structures over database schemas.
Hereafter we assume a countable domain $\U$ of elements $u, u',
\ldots$, for the interpretation of relation symbols.

\begin{definition}[Database Schema and Instance]
We call a {\em (relational) database schema} $\D$ a finite set $\{ P_1 /
q_1,\dots,P_m / q_m \}$ of relation symbols $P$ with arity $q \in
\mathbb{N}$.
%
\label{dbinstance}
Given database schema $\D$ and domain $\U$, a {\em $\D$-instance} over
$\U$ is a mapping $D$ associating each relation symbol $P \in \D$ with
a {\em finite} $q$-ary relation over $\U$, i.e., $D(P) \underset{{\tiny
f\!in}}{\subset} \U^{q}$.
\end{definition}

By Definition~\ref{dbinstance} a database instance is a finite (relational)
model of a database schema.  The {\em active domain} $\adom(D)$ of an
instance $D$ is the set of all individuals in $\U$ occurring in some
tuple $\vec{u}$ of some predicate interpretation $D(P)$,
that is, $\adom(D) = \bigcup_{P \in \D} \{ u \in \U \mid u = u_i \text{
  for some } \vec{u} \in D(P) \}$.  Observe that, since $\D$ contains
a finite number of relation symbols and each $D(P)$ is finite, so is
$\adom(D)$. We denote the set of all instances over $\D$ and $\U$ as
$\D(\U)$. Clearly, the formal framework for databases we adopt is quite
simple, but still it is powerful enough to cover practical cases of
interest \cite{MaierUV84}. Here we do not discuss the pros and cons of the
relational approach to database theory and refer to the literature for
further details \cite{AbiteboulHV95}.

\begin{example} \label{ex11}
To illustrate the notions introduced above, consider a database schema
$\D_F$ for a faculty $F$, registering data on students and staff in
two ternary relations $\mathit{Students}/3$ and $\mathit{Staff}/3$,
 that
 register IDs,
names, and departments of students and staff respectively.
A database instance $D_F$ of $\D_F$ can be given, for example, as follows:\\

{\footnotesize
\begin{tabular}{|l|l|l|}
\hline
\multicolumn{3}{|l|}{$\mathit{Students}$}\\
\hline
ID & Name & Department\\
\hline
\hline
10 & Steve & History\\
11 & Carole & Computer Science\\
12 & Derek & Mechanical Engineering\\
\hline
\end{tabular}
\hspace{1cm}
\begin{tabular}{|l|l|l|}
\hline
\multicolumn{3}{|l|}{$\mathit{Staff}$}\\
\hline
ID & Name & Department\\
\hline
\hline
01 & Rose & Mechanical Engineering\\
02 & Audrey & Mechanical Engineering\\
03 &  Karl & History\\
\hline
\end{tabular}\\
}
\end{example}

To specify the properties of databases, we make use of first-order
logic with equality and no function symbols.  Let $V$ be a countable
set of {\em individual variables}, which are the only terms in the
language for the time being.
\begin{definition}[FO-formulas over $\D$]\label{def:fo}
Given a database schema $\D$, the formulas $\varphi$ of the
first-order language $\L_{\D}$ are defined by the following BNF:
\begin{eqnarray*}
\varphi & ::= & x = x'\mid P(x_1, \ldots ,x_{q}) \mid \lnot \varphi
\mid \varphi \to \varphi \mid \forall x \varphi
\end{eqnarray*}
where $P \in \D$, $x_1, \ldots ,x_{q}$ is a $q$-tuple of variables and $x,
x' $ are variables.
\end{definition}
We assume ``$=$'' to be a special binary predicate with fixed obvious
interpretation. By Def.~\ref{def:fo}, $\L_\D$ is a first-order
language with equality over the relational vocabulary $\D$ and with no
function symbols.
%
In the following we use the standard abbreviations $\exists$, 
$\wedge$,
$\vee$, and $\neq$.
Also, free and bound variables are defined as standard.  For a formula
$\varphi \in \L_{\D}$,
we write $\varphi(x_1,\ldots,x_\ell)$, or simply $\varphi(\vec x)$, to
list  in arbitrary order all free variables
$x_1,\ldots,x_\ell$ of $\varphi$.
A {\em sentence} is a formula with no free variables.  Notice that the
only terms in our language $\L_{\D}$ are individual variables. We can
add constants for individuals with some minor technical changes to the
definitions and results in the paper. However, these do not impact on
the theoretical contribution and we prefer to keep notation lighter.

To interpret FO-formulas on database instances, we introduce 
{\em assignments} as functions $\sigma: V \mapsto \U$.
Given an assignment $\sigma$, we denote by $\sigma^x_u$ the 
assignment such that
\myi $\sigma^x_u(x) = u$; and \myii 
$\sigma^x_u(x') = \sigma(x')$, for every variable $x' \in V$ different
from $x$.
We can now define the semantics of $\L_\D$.
\begin{definition}[Satisfaction]\label{def:fo-sem}
Given a $\D$-instance $D$, an assignment $\sigma$, and an FO-formula
$\varphi\in\L_{\D}$, we inductively define whether $D$ \emph{satisfies
  $\varphi$ under $\sigma$}, or $ (D, \sigma) \models \varphi$, as
follows:
\begin{tabbing}
 $ (D, \sigma)\models P(x_1,\ldots,x_{q})$ \ \ \ \=
 iff \ \ \= $\langle \sigma(x_1),\ldots,\sigma(x_{q}) \rangle \in D(P)$\\ 
 $ (D, \sigma)\models x = x'$ \> iff \> $\sigma(x)=\sigma(x')$\\
 $ (D, \sigma)\models \lnot\varphi$ \> iff \> $(D, \sigma) \not \models\varphi$\\
 $ (D, \sigma)\models \varphi \to \psi$ \> iff \> $(D,\sigma) \not \models \varphi $ or $ (D, \sigma)\models\psi$\\
 $ (D, \sigma)\models \forall x\varphi$ \> iff \> for every
 $u\in \adom(D)$, $ (D, \sigma^x_u) \models\varphi$
\end{tabbing}

A formula $\varphi$ is {\em true} in $D$, written $D\models\varphi$,
iff $ (D, \sigma) \models \varphi$, for all assignments $\sigma$.
\end{definition}

\noindent
Observe that in Def.~\ref{def:fo-sem} we adopt an {\em active-domain}
semantics, that is, quantified variables range only over the active
domain of $D$.  This is standard in database
theory \cite{AbiteboulHV95}, where $\adom(D)$ is assumed to be the
``universe of discourse''.\\

\textbf{Integrity Constraints.}
It is well-known that several properties and constraints on databases
can be expressed as FO-sentences. Here we consider some of them for
illustrative purposes.

A {\em functional dependency} is an expression of type $\ell_1, \ldots,
\ell_k \mapsto \ell_{k+1}, \ldots, \ell_{q}$.  A database instance $D$
satisfies a functional dependency $\ell_1, \ldots, \ell_k \mapsto \ell_{k+1},
\ldots, \ell_{q}$ for predicate symbol $P$ with arity $q$ iff for every
$q$-tuples $\vec{u}$, $\vec{u}'$ in $D(P)$, whenever $u_i = u'_i$ for all $i \leq k$,
then we also have $u_i = u'_i$ for all $k < i \leq q$.
If $k = 1$, we say that it is a {\em key dependency}.
Clearly, any database instance $D$ satisfies a functional dependency
$\ell_1, \ldots, \ell_k \mapsto \ell_{k+1}, \ldots, \ell_{q}$ iff it satifies the
following FO-sentence:
\begin{eqnarray*}
 \forall \vec{x} , \vec{y} \left(P(\vec{x}) \land P(\vec{y}) \land
 \bigwedge_{i \leq k }(x_i = y_i) \to \bigwedge_{k < i \leq
   q}(x_i = y_i)\right)
\end{eqnarray*}

A {\em value constraint} is an expression of type $n_k \in P_v$, where $D(P_v)$ contains a list of admissible values.  A
database instance $D$ satisfies a value constraint $n_k \in P_v$ for predicate symbol $P$ with arity $q \geq k$ iff for every $q$-tuple $\vec{u}$ in $D(P)$, $u_k \in D(P_v)$.
Also for value constraints, it is easy to check that an instance $D$
satisfies constraint $n_k \in P_v$ for symbol $P$ iff it satisfies the following:
\begin{eqnarray*}
\forall x_1, \dots, x_q (P(x_1,\dots,x_q) \to P_v(x_k))
\end{eqnarray*}

A {\em referential constraint} enforces the foreign key of a predicate
$P_1$ to be the primary key of predicate~$P_2$.  A database instance
satisfies a referential constraint on the last $k$ attributes, denoted
as $(P_1\to P_2, k)$, if for every $q_{1}$-tuple $\vec{u}\in D(P_1)$,
there exists a $q_{2}$-tuple $\vec{u}'\in D(P_2)$ such that for all
$1\leq j \leq k$ we have that $u_{q_1-k+j}=u_j'$.
A referential constraint can also be translated as an FO-sentence:
\begin{eqnarray*}
\forall \vec{x} \left( P_1(\vec{x}) \rightarrow \exists \vec{y} \left(P_2(\vec{y}) \wedge \bigwedge_{j=1}^{k} (x_{q_1-k+j}=y_j)\right) \right)
\end{eqnarray*}

\begin{example} \label{ex12}
Clearly, in the database instance in Example~\ref{ex11} there is a key
dependency between IDs and the other attributes in relations
$\mathit{Students}$ and $\mathit{Staff}$, as it is to be expected
from any well-defined notion of ID.  On the other hand, in relation
$\mathit{Staff}$ the department is not a key, as two different
tuples have ``Mechanical Engineering'' as value for this attribute.
\end{example}

%% file: aggregators.tex
The main research question we investigate in this paper concerns how to
define an aggregated database instance from the instances of
$\N=\{1,\dots,n\}$ agents.  This question is typical in social choice
theory, where judgements, preferences, etc., are aggregated according
to some notion of rationality that will be introduced in
Section~\ref{sec:collectiverationality}.

For the rest of the paper we fix a database schema $\D$ over a common
domain $\U$, and consider a {\em profile } $\vec{D} = (D_1, \dots,
D_n)$ of $n$ instances over $\D$ and $\U$.  Then, we define what is an
aggregation procedure on such instances.
\begin{definition}[Aggregation Procedure] \label{aggregator}
Given database schema $\D$ and domain $\U$, an {\em aggregation
  procedure} $F: \D(\U)^n \to \D(\U)$ is a function assigning to each
  tuple $\vec{D}$ of instances for $n$ agents an aggregated instance
  $F(\vec{D}) \in \D(\U)$.

Let $\F$ be the class of all aggregation
  procedures.
\end{definition}

We use $N_{\vec{u}}^{\vec{D}(P)} :: = \{ i \in \N \mid \vec{u} \in
D_i(P) \}$ to denote the set of agents accepting tuple $\vec{u}$ for
symbol $P$, under profile $\vec{D}$. Note that considering a unique
domain $\U$ is not really a limitation of the proposed approach:
instances $D_1, \ldots, D_n$, each on a possibly different domain
$\U_i$, for $i \leq n$, can all be seen as instances on the union
$\bigcup_{i \in \N} \U_i$ of all domains.

Hereafter we illustrate and discuss some examples of aggregation
procedures. We begin with the class of quota rules, inspired by their homonyms in judgment
aggregation \cite{DietrichListJTP2007}. This class includes the classical majority rule, as well 
the union and the intersection rules which are
well-known in modal epistemic logic, corresponding
to distributed knowledge and ``everybody knows
that'' \cite{Hintikka1962}.

\smallskip
 \textbf{Union Rule} (or nomination): for every $P \in \D$, $F(\vec{D})(P)=\bigcup_{i
  \leq n} D_i(P) $. 
Intuitively, every agent is seen as having partial but correct
  information about the state of the world. Union can be considered a
  good aggregator if databases represent the agents' knowledge bases
  (certain information).
\smallskip

\textbf{Intersection Rule} (or unanimity): for every $P \in \D$,
  $F(\vec{D})(P)=\bigcap_{i \leq n} D_i(P) $. 
Here every agent is supposed to have a partial and possibly incorrect
  vision of the state of the world.
\smallskip

\textbf{Quota Rules}: a {\em quota} rule is an aggregation rule $F$
  defined via functions $q_P : \U^q \to \{0,1, \ldots , n+1 \}$,
  associating each symbol $P$ and $q$-uple with a quota, by
  stipulating that $\vec{u} \in F(\vec{D})(P)$ iff
  $|\{i \mid \vec{u}\in D_i(P)\}| \geq q_P(\vec{u})$.  $F$ is called
  {\em uniform} whenever $q$ is a constant function for all tuples and
  symbols.  Intuitively, if a tuple $\vec{u}$ appears in at least
  $q_P(\vec{u})$ of the initial databases, then it is accepted for
  symbol $P$.

\smallskip

  
The (strict) majority rule is a (uniform) quota rule for $q = \lceil
(n+1)/2
\rceil$; while  union and intersection are quota rule for $q = 1$ and $q = n$ respectively. We call the uniform quota rules for $q = 0$ and $q =
n + 1$ {\em trivial rules}.

The literature on belief merging has proposed and studied extensively
procedures based on
distances \cite{Konieczny2000,KoniecznyPinoPerezJLC2002,KoniecznyLangMarquisAIJ2004},
and some of these rules have also been proposed in judgment
aggregation \cite{MillerOsherson2009}. We mention below one of the
archetypal rules in this class, which makes use of the symmetric
distance.

\smallskip
\textbf{Distance-Based Rule}: 
%
%
\begin{eqnarray*}
F(\vec{D}) = \argmin_{D\in \mathcal D-\text{\it instances}} \sum_{i \in \N}\sum_{P\in\mathcal D} (|D_i(P)\setminus D(P)| + |D(P) \setminus D_i(P)|)
\end{eqnarray*}

\noindent
Intuitively, the symmetric distance minimizes the ``distance'' between
the aggregated database $F(\vec{D})$ and each $D_i$, defined as the
number of tuples in $D_i$ but not in $F(\vec{D})$, plus the number of
tuples in $F(\vec{D})$ but not in $D_i$, calculated across all
$i \in \N$.

Computing the result of distance-based rules is typically a hard computational problem: for instance, the above version on arbitrary propositional constraints is a $\Theta_2^p$-complete problem \cite{KoniecznyEtAlKR2002}. 
Tractable versions can however be obtained
by restricting the minimisation to the databases obtained in the input
profiles, a viable solution when the set of individual agents in the
input is sufficiently large.  These rules are known in the literature
on judgment aggregation as \emph{most representative voter
rules} \cite{EndrissGrandiAAAI2014}, and we state here the simplest
one.

\textbf{Average Voter Rule}: 
\begin{eqnarray*}
F(\vec{D}) & = & \argmin_{\{D_i\mid i\in\N\}} \sum_{j \in \N} \sum_{P\in\mathcal D} (|D_j(P)\setminus D_i(P)| + |D_i(P) \setminus D_j(P)|)
\end{eqnarray*}


Observe the the two rules above might output a set of equally
preferred extensions for a relation symbol $P$. i.e., they are \emph{non-resolute} rules.
%
We also consider a slight variant of the average voter rule.

\textbf{Relation-wise Average Voter Rule}: 
\begin{eqnarray} \label{ag1}
F(\vec{D}) & = & \bigcup_{P \in\mathcal D} \argmin_{\{D_i(P) \mid
i\in\N\}} \sum_{j \in \N} (|D_j(P)\setminus D_i(P)| +
|D_i(P) \setminus D_j(P)|)
\end{eqnarray}

Notice that, according to (\ref{ag1}), the average is computed for
each $P \in \D$ independently. In particular, $F(\vec{D})$ does not
correspond to any of $D_1, \ldots, D_n$ in general.

We now state a class of rules which are typically considered 
non-desirable in the literature on social choice theory, since
they leave a somewhat large set of agents out of the aggregation.

\smallskip
\textbf{Dictatorship of Agent $i^* \in \N$}: 
we have that $F(\vec{D}) = D_{i^*}$, i.e., the dictator $i^*$
  completely determines the aggregated database.
\smallskip

{\bf Oligarchy of Coalition $C^* \subseteq \N$}: for every $P \in \D$,
  $F(\vec{D})(P) = \bigcap_{i \in C^*} D_{i}(P) $.  Oligarchy reduces
  to dictatorship for singletons, and to intersection for $C^* = \N$.
\smallskip

To conclude, we present a novel definition of aggregation procedure
inspired by the literature on incomplete information in
databases \cite{Libkin15}.

\smallskip
{\bf Merge with Incomplete Information}: for every $P \in \D$,
  $\vec{u} \in F(\vec{D})(P)$ iff (i) for some $\vec{u}_1 \in
  D_{1}(P)$, \ldots, $\vec{u}_n \in D_{n}(P)$, for every $k \leq q$,
  either $u_{1,k} = \ldots = u_{n,k}$ and $u_k = u_{1,k}$, or
  $u_{j,k} \neq u_{j'',k}$ for some $j, j'' \leq n$, and $u_j = \bot$,
  where $\bot$ is a new symbol; (ii) for every $\vec{u}' \in
  F(\vec{D})(P)$ and $k \leq q$, $u_k = \bot$ implies $u'_k = \bot$ or
  for some $k \leq q$, $u_k \neq u'_k$.
\smallskip

That is, by (i) whenever elements $u_1, \ldots, u_k$ appear at the
same positions in some tuples in the profile, then they will appear at
those positions in $F(\vec{D})$. On the other hand, if different
elements appear, then we insert symbol $\bot$ as a placeholder.  By
(ii) we discard tuples with ``strictly less'' information.  Notice
that merge with incomplete information does not conform entirely with
Def.~\ref{aggregator}, as the outcome $F(\vec{D})$ is a database
instance on $\D(\U \cup \{\bot\})$, rather than $\D(\U)$. Nonetheless
all relevant notions on databases and aggregators can be extended
seamlessly in what follows.



\begin{example} \label{ex13}
Suppose that the database instance $D_F$ in Example~\ref{ex11} is
owned by the HR department of the faculty. On the other hand, the
registrar and the head office own the following instances $D'_F$ and $D''_F$ respectively, due to
differences in updating mechanisms and possibly errors:

\hspace{-0.9cm}
{\footnotesize
\begin{tabular}{cc}

\begin{tabular}{|l|l|l|l|l|l|}
\hline
\multicolumn{3}{|l|}{$\mathit{Students}$} & \multicolumn{3}{|l|}{$\mathit{Staff}$}\\
\hline
ID & Name & Department & ID & Name & Department\\
\hline
\hline
10 & Steve & History & 01 & Rose & Mech.~Eng.\\
11 & Carole & CS & 02 & Audrey & Mech.~Eng.\\
 &  & & 04 &  Carl & History\\
\hline
\end{tabular}

\begin{tabular}{|l|l|l|l|l|l|}
\hline
\multicolumn{3}{|l|}{$\mathit{Students}$} & \multicolumn{3}{|l|}{$\mathit{Staff}$}\\
\hline
ID & Name & Department & ID & Name & Department\\
\hline
\hline
10 & Steve & History & 01 & Rose & Mech.~Eng.\\
11 & Carole & CS & 02 & Aubrey & Mech.~Eng.\\
12 & Derek & Mech.~Eng. & 03 &  Karl & History\\
13 & Marc  & History &  &  &\\
\hline
\end{tabular}
\end{tabular}
}
\medskip

\noindent
To provide a unique vision of instances $D_F$, $D'_F$, and $D''_F$ we
can in principle choose any of the aggregation procedures introduced
above. For instance, the intersection, union and majority aggregated
profiles can be given as follows:\\

{\footnotesize
\begin{tabular}{cc}
\textbf{Intersection}:&
\begin{tabular}{|l|l|l|l|l|l|}
\hline
\multicolumn{3}{|l|}{$\mathit{Students}$} & \multicolumn{3}{|l|}{$\mathit{Staff}$}\\
\hline
ID & Name & Department & ID & Name & Department\\
\hline
\hline
10 & Steve & History & 01 & Rose & Mech.~Eng.\\
11 & Carole & CS & & & \\
\hline
\end{tabular}\\
& \\

\textbf{Union}: &
\begin{tabular}{|l|l|l|l|l|l|}
\hline
\multicolumn{3}{|l|}{$\mathit{Students}$} & \multicolumn{3}{|l|}{$\mathit{Staff}$}\\
\hline
ID & Name & Department & ID & Name & Department\\
\hline
\hline
10 & Steve & History & 01 & Rose & Mech.~Eng.\\
11 & Carole & CS & 02 & Audrey & Mech.~Eng.\\
12 & Derek & Mech.~Eng. & 02 &  Aubrey & Mech.~Eng.\\
13 & Marc  & History & 03 &  Karl & History\\
 &  &  & 04 &  Carl & History\\
\hline
\end{tabular}\\

& \\

\textbf{Majority}:&
\begin{tabular}{|l|l|l|l|l|l|}
\hline
\multicolumn{3}{|l|}{$\mathit{Students}$} & \multicolumn{3}{|l|}{$\mathit{Staff}$}\\
\hline
ID & Name & Department & ID & Name & Department\\
\hline
\hline
10 & Steve & History & 01 & Rose & Mech.~Eng.\\
11 & Carole & CS & & & \\
12 & Derek & Mech.~Eng. & & & \\
\hline
\end{tabular}
\end{tabular}
}
\medskip

Clearly, some aggregation procedure do not preserve all integrity constraints, e.g., unions do not preserve key dependencies.
Furthermore, aggregation by the average voter rule would output $D_F$, and by merge with incomplete information would produce the following instance:\\

{\footnotesize
\begin{tabular}{cc}
\textbf{Merge with inc. information:} &
\begin{tabular}{|l|l|l|l|l|l|}
\hline
\multicolumn{3}{|l|}{$\mathit{Students}$} & \multicolumn{3}{|l|}{$\mathit{Staff}$}\\
\hline
ID & Name & Department & ID & Name & Department\\
\hline
\hline
10 & Steve & History & 01 & Rose & Mech.~Eng.\\
11 & Carole & CS & 02 & $\bot$ & Mech.~Eng.\\
 &  &  & $\bot$ & $\bot$ & History\\
\hline
\end{tabular}
\end{tabular}
}
\medskip

\noindent
In particular, in this last instance symbol $\bot$ intuitively signals
that we are uncertain about the name of staff with ID 02, as well as
about the ID and name of staff in the History department.
\end{example}

Note that further aggregation procedures are possible in principle.
We choose to focus on those above
as they are inspired by well-studied procedures from the literature
and, with the exception of the distance-based rule, they are tractable
computationally.

%% file: axioms.tex
Aggregation procedures are best characterised by means of axioms. In
particular, we consider the following properties, where relation
symbols $P, P' \in \D$, profiles $\vec{D}, \vec{D}' \in \D(\U)^n$,
tuples $\vec{u}$, $\vec{u}' \in \U^+$ are all universally quantified.
We leave the treatment of the merge with incomplete information rule
for the end of the section.
%

\smallskip

\textbf{Unanimity ($U$)}: 
$F(\vec{D})(P) \supseteq \bigcap_{i \in \N} D_i(P)$.

\smallskip

By unanimity a tuple accepted by all agents also appears in the
aggregated database (for the relevant relation symbol). With the
exception of the distance-based rule and trivial quota rules with any
of the $q_P=n+1$, the remaining aggregators from
Section~\ref{aggregators} all satisfy unanimity.
%

\smallskip

\textbf{Groundedness ($G$)}: 
$F(\vec{D})(P) \subseteq \bigcup_{i \in \N}
  D_i(P)$.

\smallskip

By groundedness any tuple appearing in the aggregated database must be
accepted by some agent. All aggregators in Section~\ref{aggregators},
with the exception of the distance-based rule and the trivial quota
rule with any of the $q_P=0$, satisfy this property.

\smallskip

\textbf{Anonymity ($A$)}: for every 
permutation $\pi : \N \to \N$, we have $F(D_1, \ldots , D_n ) = F
(D_{\pi(1)} , \ldots , D_{\pi(n)})$.

\smallskip

By anonymity the identity of agents is irrelevant for the aggregation
procedure. This is the case for all aggregators in
Section~\ref{aggregators} but dictatorship and oligarchy.

\smallskip

\textbf{Independence ($I$)}: 
if $N_{\vec{u}}^{\vec{D}(P)} = N_{\vec{u}}^{\vec{D}'(P)}$ then
$\vec{u} \in F(\vec{D})(P)$ iff $\vec{u} \in F(\vec{D}')(P)$.

\smallskip

Intuitively, if the same agents accept (resp.~reject) a tuple in two
different profiles, then the tuple is accepted (resp.~rejected) in
both aggregated instances.  The axiom of independence is a widespread
requirement from social choice theory, and is arguably the main cause
of most impossibility theorems, such as Arrow's seminal
result \cite{Arrow1963}.  From a computational perspective,
independent rules are typically easier to compute than non-independent
ones. Quota rules satisfy independence, as well as dictatorships and oligarchies.
%




\smallskip

\textbf{Positive Neutrality ($N^+$)}: 
if $N_{\vec{u}}^{\vec{D}(P)} = N_{\vec{u}'}^{\vec{D}(P)}$ then
$\vec{u} \in F(\vec{D})(P)$ iff $\vec{u}' \in F(\vec{D})(P)$.
\smallskip

\textbf{Negative Neutrality ($N^-$)}: 
if $N_{\vec{u}}^{\vec{D}(P)} = \N \setminus N_{\vec{u}'}^{\vec{D}(P)}$ then
$\vec{u} \in F(\vec{D})(P)$ iff $\vec{u}' \not\in F(\vec{D})(P)$.
\smallskip

Observe that both versions of neutrality differ from independence as
here we consider two different tuples in the same profile, while
independence deals with the same tuple in two different profiles.  Again, with the exception of the distance-based rule all aggregators introduced in
Section~\ref{aggregators} satisfy positive neutrality.
Most quota rules including union and intersection do not satisfy negative neutrality (see
Lemma~\ref{lemma:quotaaxioms} below), 
dictatorships and oligarchies satisfy the latter axiom.
%








\smallskip

\textbf{Permutation-Neutrality ($N^{P}$)}: Let $\rho
  : \U \to \U$ be a permutation over the domain $\U$, and $\rho(\vec{D})$ its straightforward lifting to a profile $\vec{D}$, then 
  $F(\rho(\vec{D}))=\rho(F(\vec{D}))$.

\smallskip

All aggregators introduced in Section~\ref{aggregators} satisfy
permutation-neutrality.
%
%
We conclude with the following axiom, that formalises the fact that an aggregator keeps on accepting a given tuple if the support for that tuple increases.
 
\smallskip

\textbf{Monotonicity ($M$)}: 
if $\vec u \in F(\vec{D})(P)$ and for every $i \in \N$, either
$D_i(P) = D'_i(P)$ or $D_i(P) \cup \{\vec{u}\} \subseteq D'_i(P)$, then $\vec
u \in F(D')(P)$.
\smallskip





\smallskip

Combinations of the axioms above can be used to characterise some of
the rules that we defined in Section~\ref{aggregators}. Some of these
results, such as the following, lift to databases known results in
judgement (propositional) aggregation.

\begin{lemma} \label{lemm1}
An aggregation procedure satisfies $A$, $I$, and $M$ 
iff it is a quota rule.
\end{lemma}
\begin{proof}[Proof sketch]
The right-to-left implication follows from the fact that quota rules
satisfy independence $I$, anonymity $A$, and monotonicity $M$, as we
remarked above.
For the left-to-right implication, observe that, to accept a
given tuple $\vec{u}$ in $F(\vec{D})(P)$, an independent aggregation
procedure will only look at the set of agents $i \in \N$ such that
$\vec{u} \in D_i(P)$. If the procedure is also anonymous, then
acceptance is based only on the number of individuals admitting the
tuple. Finally, by monotonicity, there is a minimal number of
agents required to trigger collective acceptance. That number is the
quota associated with the tuple and the symbol at hand.
\end{proof}


If we add neutrality,
then we obtain the class of uniform quota rules. If we furthermore
impose unanimity and groundedness, we exclude the trivial quota rules.
\begin{lemma}\label{lemma:quotaaxioms}
If the number of individuals is odd and $|\D| \geq 2$, an aggregation
procedure $F$ satisfies $A$, $N^-$, $N^+$, $I$ and
$M$ on the full domain $\D(\U)^n$ if and only if it is the majority
rule.
\end{lemma}
\begin{proof}
By positive neutrality the quota must be the same for all tuples and all
relation symbols. By negative neutrality the two sets
$N_{\vec{u}}^{\vec{D}(P)}$ and $\N \setminus N_{\vec{u}}^{\vec{D}(P)}$
must be treated symmetrically.  Hence, the only possibility is to have
a uniform quota of $(n +1)/2$.
\end{proof}

The corresponding versions of these results have been proved to hold
in judgment and graph
aggregation \cite{DietrichListJTP2007,EndrissGrandiAIJ2017}.
%
We now show the following equivalence between majority and the distance-based rule in the absence of  integrity constrains:
\begin{lemma}
In absence of constraints, and for an odd number of
agents, the distance-based rule coincides with the majority rule.
\end{lemma}

\begin{proof}
%
With a slight abuse of notation, if $A \subseteq \U^{m}$ let
$A(\vec{u})$ be its characteristic function.  Recall the definition of
the distance-based rule.  Since the minimisation is not constrained,
and all structures are finite, the definition is equivalent to the
following:
\begin{eqnarray*}
F(\vec{D})(P) & = &\argmin_{D\in \mathcal D-\text{\it instances}} \sum_{i \in \N} \sum_{P\in\mathcal D}  \sum_{\vec{u}\in \U^{q_P}} (|D_i(P)(\vec{u})- D(P)(\vec{u})|)
 \\
& = & \argmin_{D\in \mathcal D-\text{\it instances}}  \sum_{P\in\mathcal D}  \sum_{\vec{u}\in \U^{q_P}} \sum_{i \in \N} (|D_i(P)(\vec{u})- D(P)(\vec{u})|)
\end{eqnarray*}
Therefore, for each $P\in\mathcal D$ and for each $\vec{u}$, if for a majority of the individuals in
$\N$ we have that $\vec{u}\in D_i(P)$, then $\vec{u}\in D(P)$ minimises
the overall distance, and symmetrically for the case in which a
majority of individuals are such that $\vec{u}\not \in
D_i(P)$. \end{proof}

It is easy to see that the above lemma does not hold in presence of arbitrary integrity constraints (see Example~\ref{ex:majority}).
We conclude this section with a result on merge with incomplete information. Note that the axioms need to be slightly adapted to account for the addition of a symbol in the output of the aggregator.
\begin{lemma}
Merge with incomplete information satisfies $U$, $A$, $I$, and $N^+$,
but not $N^-$, nor $M$. In particular, by Lemma~\ref{lemm1} it is not
a quota rule.
\end{lemma}
\begin{proof}
 The proof that merge satisfies unanimity $U$ is immediate, similarly
 for anonymity $A$. As regards independence $I$, if
 $N_{\vec{u}}^{\vec{D}(P)} = N_{\vec{u}}^{\vec{D}'(P)} = \N$ then
 $\vec{u} \in F(\vec{D})(P)$ and $\vec{u} \in F(\vec{D}')(P)$ by
 unanimity. On the other hand, if $N_{\vec{u}}^{\vec{D}(P)} =
 N_{\vec{u}}^{\vec{D}'(P)} \neq \N$ then $\vec{u} \notin
 F(\vec{D})(P)$ and $\vec{u} \notin F(\vec{D}')(P)$.  For positive
 neutrality the reasoning is similar: if $N_{\vec{u}}^{\vec{D}(P)} =
 N_{\vec{u}'}^{\vec{D}(P)} = \N$ then $\vec{u} \in F(\vec{D})(P)$ and
 $\vec{u}' \in F(\vec{D})(P)$; whereas if $N_{\vec{u}}^{\vec{D}(P)} =
 N_{\vec{u}'}^{\vec{D}(P)} \neq \N$ then $\vec{u} \notin
 F(\vec{D})(P)$ and $\vec{u}' \notin F(\vec{D})(P)$.  Finally, it is
 not difficult to find counterexamples for both negative neutrality
 $N^-$ and monotonicity $M$.
 For instance, as regards $N^-$ consider Example~\ref{ex13} and tuples $(02, Audrey, Mech.~Eng.)$ and $(02, Aubrey, Mech.~Eng.)$ in relation $\mathit{Staff}$.

\end{proof}


%% file: collectiverat.tex
%

In this section we analyse further the properties of the aggregation
procedures introduced in Section~\ref{aggregators}. First, we
present a notion of {\em collective rationality} that aims to capture
the appropriateness of a given aggregator $F$ w.r.t.~some constraint
$\phi$ on the input instances $D_1, \ldots, D_n$.
Hereafter let $\phi$ be a sentence in the first-order language
$\L_{\D}$ associated with $\D$, interpreted as an integrity constraint
that is satisfied by all $D_1, \ldots, D_n$.

%
\begin{definition}[Collective Rationality]
A constraint $\phi$ is lifted by an aggregation procedure $F$ if
whenever $D_i \models \phi$ for all $i \in \N$, then also
$F(\vec{D}) \models \phi$. 
An aggregation procedure $F :
\D(\U)^n \to \D(\U)$ is {\em collectively rational} (CR) with respect to
$\phi$ iff $F$ lifts $\phi$.
\end{definition}

Intuitively, an aggregator is CR w.r.t.~constraint $\phi$ iff it
lifts, or preserves, $\phi$. Consider the following:
\begin{example}\label{ex:majority}
By Example~\ref{ex13} unions are not collective rational
w.r.t. dependency constraints. We also provide a further example of
first-order collective (ir)rationality with the majority rule.
Consider agents 1 and 2 with database schema $\D = \{ P/1, Q/2 \}$.
Two database instances are given as $D_1 = \{ P(a), Q(a,b) \}$ and
$D_2 = \{ P(a), Q(a,c) \}$. Clearly, both instances satisfy
integrity constraint $\phi = \forall x (P(x) \to \exists y Q(x,y))$.
However, their aggregate $D = F(D_1,D_2) = \{ P(a) \}$, obtained by
the majority rule, does not satisfy $\phi$.
These small examples, which can be considered a {\em paradox} in the
sense of \cite{GrandiEndrissAIJ2013}, shows that not every constraint
in the language $\L_{\D}$ is collective rational w.r.t.~unions and
majority, thus obtaining a first, simple negative result.
\end{example}

We now focusing on integrity constraints that are proper to databases,
as defined in Section~\ref{sec:preliminaries}, presenting sufficient
(and possibly necessary) conditions for aggregators to lift them.  We
begin with functional dependencies.

\begin{proposition} \label{prop1}
A quota rule lifts a functional constraint iff for
all relation symbols $P$ occurring in the functional constraint
we have that $q_P > \frac{n}{2}$, where $n$ is the number of agents.
\end{proposition}

\begin{proof}
By assumption, every instance $D_i$ satisfies the constraint.  That is
for every tuple $(u_1,\dots, u_k)$, either there is a unique
$(u_{k+1},\dots, \allowbreak u_q)$ such that
$(u_1,\dots,u_q)=\vec{u}\in D_i(P)$, or there is none.  Suppose now
that the constraint is falsified by the collective outcome.  That is,
there are $\vec{u} \neq \vec{u}'$ such that both $\vec{u}\in
F(\vec{D})(P)$ and $\vec{u}'\in F(\vec{D})(P)$, and
$\vec{u}$ and $\vec{u}'$ coincide on the first $k$ coordinates.  By
definition of quota rules, this means that at least $q_P$ voters are
such that $\vec{u}\in D_i(P)$, and at least $q_P$ possibly different
voters had $\vec{u}'\in D_i(P)$.  Since each individual can have
either $\vec{u}$ or $\vec{u}'$ in $D_i(P)$, by the pigeonhole
principle this is possible if and only if the quota
$q_P \leq \frac{n}{2}$.
\end{proof} 

As immediate applications of Prop.~\ref{prop1}, the intersection rule
clearly lifts any functional dependency, while the union lifts none,
as previously illustrated.
%
\begin{proposition}\label{prop2}
An aggregation procedure $F$ lifts a value constraint if $F$ is
grounded.
\end{proposition}
\begin{proof}
Let $n_k\in D(P_v)$ be a value constraint, where for all $i,j\in \N$,
we have that $D_i(P_v)=D_j(P_v)$.  A grounded aggregation procedure is
such that $F(\vec{D})(P)\subseteq \bigcup_{i \in \N} D_i(P)$. Hence,
for all $\vec{u}\in F(\vec{D})(P)$, there exists an $i\in \N$ such
that $\vec{u}\in D_i(P)$. Since all individual databases satisfy the
value constraint, we have that $u_k\in D_i(P_v)$, and therefore
$u_k\in F(\vec{D})(P_v)\subseteq \bigcup_{i \in \N} D_i(P_v)$, showing
that also $F(\vec{D})(P)$ satisfies the value constraint.
\end{proof}

The converse of the Prop.~\ref{prop2} is not true in general, since a
non-grounded aggregator could be easily devised while still satisfying
a given value constraint.
Finally, we consider referential constraints.
%
\begin{proposition} \label{prop3}
A quota rule lifts a referential constraint $(P_1\to P_2, k)$ iff
$q_{P_2} = 1$.
\end{proposition}

\begin{proof}
Let $\vec{u}\in F(\vec{D})(P_1)$.  Since all the individual databases
satisfy the integrity constraint, we know that for every $i \in \N$
there exists a $\vec{u}_i\in D_i(P_2)$ such that its first $k$
coordinates coincides with the last $k$ coordinates of $P_1$.  Since
all $\vec{u}_i$ are possibly different, they may be supported by one
single individual each. Therefore, the referential constraint is lifted
if and only if the quota relative to $P_2$ is sufficiently small,
i.e., $q_{P_2} = 1$.
\end{proof}

As an immediate application of Prop.~\ref{prop3}, intersection and union
rules are included in the results above, since they are quota
rules. As for the distance-based rule or the average voter rule, we only remark that they lift
all integrity constraint by their definition, provided that the
minimisation is restricted to consistent databases.

For simpler integrity constraints, notably conjunctions of literals,
we show a simple correspondence theorem in the spirit of
\cite{GrandiEndrissAIJ2013}, albeit adapted to the first-order language under consideration.
First, we introduce a set $Con \subseteq \U$ of constants in
the first-order language, interpreted as themselves in each $D_i$,
that is, $\sigma(c) = c$ for every $c \in Con$.  Then, let
$\lit^+ \subseteq \L_{\D}$ be some language containing only positive
literals of form $P(c_1, \ldots, c_q)$, let 
$\lit^- \subseteq \L_{\D}$ be the set containing only {\em negative}
literals $\neg P(c_1, \ldots, c_q)$, and $\lit = lit^+ \cup lit^-$.  We can prove the following:

\begin{theorem}
If an aggregator $F$ satisfies $U$ and $G$, then it is collectively
rational w.r.t.~$\lit$.  If $Con = \U$,
then every aggregator that is collectively rational w.r.t.~$\lit$ is
also unanimous and grounded.
\end{theorem}

\begin{proof}
Suppose that aggregator $F$ satisfies $U$ and $G$. Then, we show that it
is collectively rational w.r.t.~$\lit$.
In particular, if all instances $D_1, \ldots, D_n$ satisfy formulas
$P(c_1, \ldots, c_q)$ in $lit^+$, then $\vec{c} \in D_i(P)$ for every
$i \in \N$.  By unanimity we have that $\bigcap_{i \in \N}
D_i(P) \subseteq F(\vec{D})(P)$, and therefore $\vec{c} \in
F(\vec{D})(P)$.  Hence, $F$ is collectively rational on $\lit^+$.  A
similar reasoning holds for $\lit^-$ by using groundedness, and
therefore for the whole $\lit$.

As for the converse, suppose that $Con = \U$ and $F$ is collectively
rational wrt to $\lit$. Then, choose a profile $D_1, \ldots, D_n$ with
$\vec{u} \in \bigcap_{i \in \N} D_i(P)$, that is, for every
$i \in \N$, $D_i \models P(u_1, \ldots, u_q)$. Since we assumed that
$Con=\U$, all formulas $P(u_1, \ldots, u_q)$ are in $lit^+$. Further,
$F$ is CR on $D_1, \ldots, D_n$ and therefore $F(\vec{D}) \models
P(u_1, \ldots, u_q)$, that is, $\vec{u} \in F(\vec{D})(P)$, which
means that $F$ is unanimous.  Similarly, and under the same
assumption, any $F$ that is collectively rational w.r.t.~$\lit^-$ is
grounded.
\end{proof}

Note that, differently from the propositional case \cite[Theorem
10]{GrandiEndrissAIJ2013}, here we need both axioms of unanimity and
groundedness to preserve both positive and negative literals, while
for propositional languages unanimity suffices.

Given the results above, a natural question is to identify the class of aggregators that can lift any integrity constraint, no matter its form. Let us first define the following class:

\begin{definition}[Generalised dictatorship]
 An aggregation procedure $F : \D(\U)^n \to \D(\U)$ is a {\em
   generalised dictatorship} if there exists a map $g: \D(\U)^n \to \N$
   such that for every $\vec{D} \in \D(\U)^n$, $F(\vec{D}) =
   D_{g(\vec{D})}$. Let $GDIC$ be the class of generalised
   dictatorships.
\end{definition}

Generalised dictatorships include classical dictatorships, but also
more interesting procedures such as the average voter rule from Section~\ref{aggregators}, or any other rule which selects the individual input that best summarises a
given profile.  Clearly, since each single instance satisfies the
given set of constraints, a generalised dictatorship is
collectively rational with respect to the full first-order language.
 \begin{theorem} 
$GDIC \subset CR[\L_{\D}]$ 
\end{theorem}

Observe that while for binary aggregation the theorem above is an
equality \cite[Theorem~16]{GrandiEndrissAIJ2013}, this is not the
case for database aggregation. This is due to the fact that the
first-order language specifies a given database instance only up to
isomorphism. The proof of this fact is rather immediate: consider a
dictatorship of the first agent, modified by permuting all the
elements in $\U$. That is, $F(\vec{D})=\rho(D_1)$ where
$\rho:\U\to \U$ is any permutation different from the identity.
Clearly, $D_1 \neq \rho(D_1)$
but all constraints that
were satisfied by $D_1$ are also satisfied by $\rho(D_1)$. Hence, this
aggregator is collectively rational with respect to the full
first-order language $\L_{\D}$, but is not a generalised dictatorship.


%% file: query.tex
In this section we analyse one of the most common operations performed
on databases, i.e., query answering, in the light of (rational)
aggregation.
Observe that any open formula $\phi(x_1,\dots,x_\ell)$, with free
variables $x_1,\dots,x_\ell$, can be thought of as a
query \cite{AbiteboulHV95}. Evaluating $\phi(x_1,\dots,x_\ell)$ on a
database instance $D$ returns the set $ans(D, \phi)$ of tuples
$\vec{u}=(u_1,\dots,u_\ell)$ such that the assignment $\sigma$, with
$\sigma(x_i) = u_i$ for $i \leq \ell$, satisfies $\phi$, that is,
$(D, \sigma) \models \phi$. Hereafter, with an abuse of notation, we
often write simply $(D, \vec{u}) \models \phi$.
Given the importance of query answering in database theory, the
following question is of obvious interest.
\begin{question}
What is the relation between the answer $ans(F(\vec{D}), \phi)$ to
query $\phi$ on the aggregated database $F(\vec{D})$, and answers
$ans(D_1, \phi), \allowbreak \ldots, ans(D_n, \phi)$ to the same query
on each instance $D_1, \ldots, D_n$?
\end{question}

Clearly, given a query $\phi$, every aggregator $F$ on database
instances induces an aggregation procedure $F^*$ on the query
answers,
as illustrated by the following diagram:
\begin{center}
\begin{tikzpicture}[auto, node distance=2cm, ->, >=stealth', shorten
>=1pt, semithick]

\tikzstyle{every place/.style}=[fill=white, text=black, minimum size=15pt]
\tikzstyle{every initial by arrow}=[initial text=]
\node 	(s0)   	{$D_1, \ldots, D_n$};
\node    (s1) 	[right of=s0, node distance=5cm]	{$F(\vec{D})$};
\node    (s2) 	[below of=s0]	{$ans(D_1, \phi), \ldots, ans(D_n, \phi)$};
\node    (s3) 	[below of=s1]	{$ans(F(\vec{D}),\phi)$};
\path 	(s0) 	edge node[above] {$F$} (s1);
\path 	(s0) 	 edge node[left] {$\phi$} (s2);
\path 	(s1) 	 edge node[right] {$\phi$} (s3);
\path 	(s2) 	 edge node[above] {$F^*$} (s3);
\end{tikzpicture}
\end{center}

Hereafter we consider some examples to illustrate this question.
\begin{example} \label{ex1}
If we assume intersection as the aggregation procedure, it is easy to
check that in general the answer to a query in the aggregated database
is not the intersection of the answers for each single instance. To
see this, let $D_1(P) = \{ (a,b) \}$ and $D_2(P) = \{ (a,d) \}$ and
consider query $\phi = \exists y P(x,y)$. Clearly, $ans(D_1 \cap
D_2, \phi)$ is empty, while $ans(D_1, \phi) \cap ans(D_2, \phi)
= \{a\}$.  Hence, in general $\bigcap_{i \in \N}
ans(D_i, \phi) \not \subseteq ans(\bigcap_{i \in \N} D_i, \phi)$.
The converse can also be the case. Consider instances $D_1$,
$D_2$ such that $D_1(P) = \{(a,a), (a,b)\}$, $D_1(R) = \{ c \}$, and
$D_2(P) = \{(a,a), (a,b)\}$, $D_2(R) = \{ d \}$, 
with query $\phi = \forall y P(x,y)$. The intersection
$ans(D_1, \phi) \cap ans(D_2, \phi)$ of answers is empty.
However the answer w.r.t.~the intersection of databases is
$ans(D_1 \cap D_2, \phi) = \{ a \}$, since the active domain of the
intersection only includes elements $a$ and $b$.  As a result, in
general $ans(\bigcap_{i \in \N} D_i, \phi) \not \subseteq
\bigcap_{i \in \N} ans(D_i, \phi)$.

A similar argument shows that the union of answers is in
general different from the answer on the union of instances.
\end{example}


These examples shows that it is extremely difficult to find
aggregators such that the diagram above commutes for any first-order
query $\phi \in \L_{\D}$. Hence, they naturally raise the question of
syntactic restrictions on queries such that the aggregation procedure
$F^* = \phi \circ F \circ \phi^{-1}$ on answers can be expressed
explicitly in terms of $F$ (e.g., the intersection of answers is the
answer to the query on the intersection)\footnote{Hereafter, with an abuse of notation, we write $\phi$ also to mean the corresponding query evaluation for formula $\phi$.}:
\begin{question} \label{quest1}
Given aggregation procedures $F$ and $F^*$, is there a restriction on
the query language for $\phi$ such that the diagram above commutes?
\end{question}

This problem is related to the following, more general question.
\begin{question} \label{quest2}
  Given an aggregation procedure $F$ and a query 
  language $\L$, what is the aggregation procedure $F^*$?  Can $F^*$
be represented ``explicitly'', for instance as one of the aggregation procedure introduced in Sec.~\ref{aggregators}?
\end{question}

For instance, it is immediate that if $F$ and $F^*$ are both
dictatorships for the same agent $i^* \in \N$, then the whole
first-order language $\L$ is preserved, that is, the result of
querying and then aggregating by $F^*$ is the same as aggregating by $F$ and then querying.



The results in this section provide a first, partial answer to
Question~\ref{quest1}.
%
Hereafter, with an abuse of notation, we consider functor $F^*$ as an
aggregator on databases. Indeed, all $ans(D_i, \phi)$ can be seen as a
finite relational structures, to which we can apply the aggregators in
Sec.~\ref{aggregators}, as well as the axioms in Sec.~\ref{sec:axioms}.

Let us first introduce the positive existential and universal fragments
$\L^+_{\exists}$ and $\L^+_{\forall}$ of first-order logic, defined
respectively as follows:
\begin{eqnarray*}
\phi  & ::= &  x = x' \mid P(x_1, \ldots, x_q) \mid 
\phi \lor \phi \mid \exists x \phi\\
\phi  & ::= &  x = x' \mid P(x_1, \ldots, x_q) \mid 
\phi \land \phi \mid \forall x \phi
\end{eqnarray*}

Our first lemma shows a positive result related to Example~\ref{ex1}, when union is considered instead of intersection.
\begin{lemma}[Existential Fragment] \label{existential}
The language $\L^+_{\exists}$ is preserved by unions, that is, for $F$ and $F^*$ equal to
set-theoretical union, the diagram commutes for the query language
$\L^+_{\exists}$.
\end{lemma}
\begin{proof}
The proof is by induction on the structure of query $\phi$.
For atomic $\phi = P(x_1, \ldots, x_q)$, we have that $\vec{u} \in
ans(\bigcup_{i \in \N} D_i,\phi)$ iff $(\bigcup_{i \in \N}
D_i, \vec{u}) \models \phi$, iff for some $i \in \N$, $(D_i,
\vec{u}) \models \phi$, iff $\vec{u} \in ans(D_i,\phi)$  for some $i \in \N$,
iff $\vec{u} \in \bigcup_{i \in \N} ans(D_i,\phi)$.


For $\phi = \psi \lor \psi'$, $\vec{u} \in ans(\bigcup_{i \in \N}
  D_i,\phi)$ iff $(\bigcup_{i \in \N} D_i, \vec{u}) \models \phi$, iff
  $(\bigcup_{i \in \N} D_i, \vec{u}_1) \models \psi$ or
  $(\bigcup_{i \in \N} D_i, \vec{u}_2) \models \psi'$, iff for some
  $i, j \in \N$, $(D_i, \vec{u}_1) \models \psi$ or
  $(D_j, \vec{u}_2) \models \psi'$ by induction hypothesis, where
  $\vec{u}_1$ and $\vec{u}_2$ are suitable subsequences
  of $\vec{u}$. In particular, we have both
  $(D_i, \vec{u}) \models \psi \lor \psi'$ and
  $(D_j, \vec{u}) \models \psi \lor \psi'$, that is,
  $\vec{u} \in \bigcup_{i \in \N} ans(D_i,\phi)$.  On the other hand,
  $\vec{u}
\in \bigcup_{i \in \N} ans(D_i,\phi)$ iff $\vec{u} \in ans(D_i,\phi)$
for some $i \in \N$, iff $(D_i, \vec{u}_1) \models \psi$ or $(D_i,
\vec{u}_2) \models \psi'$. In both cases, by induction hypothesis $(\bigcup_{i \in \N} D_i,
\vec{u}) \models \phi$, that is, $\vec{u} \in ans(\bigcup_{i \in \N}
  D_i,\phi)$.

For $\phi = \exists x \psi$, $\vec{u} \in ans(\bigcup_{i \in \N}
D_i,\phi)$ iff $(\bigcup_{i \in \N} D_i, \vec{u}) \models \phi$, iff
for some $v \in adom(\bigcup_{i \in \N} D_i)$, $(\bigcup_{i \in \N}
D_i, \allowbreak \vec{u} \cdot v) \models \psi$, and therefore for some $i,
j \in \N$, $v \in adom(D_j)$ and $(D_i, \vec{u} \cdot
v) \models \psi$.  Note that if $(D_i, \vec{u} \cdot
v) \models \psi$, then $v \in adom(D_i)$ as well, as $\phi$ belongs to
the positive (existential) fragment of first-order logic.
Hence, for some $i \in \N$, $v \in adom(D_i)$ and $(D_i, \vec{u} \cdot
v) \models \psi$,  that is,
$\vec{u} \in ans(D_i,\phi)$ for some $i \in \N$. 
On the other hand, $\vec{u} \in \bigcup_{i \in \N} ans(D_i,\phi)$ iff
$\vec{u} \in ans(D_i,\phi)$ for some $i \in \N$, iff $v \in adom(D_i)$
and $(D_i, \vec{u} \cdot v) \models \psi$, that is, $v \in
adom(\bigcup_{i \in \N} D_i)$ and $(\bigcup_{i \in \N}
D_i, \vec{u} \cdot v) \models \psi$ by induction hypothesis. Hence,
$\vec{u} \in ans(\bigcup_{i \in \N} D_i,\phi)$.
\end{proof}

By Lemma~\ref{existential} queries in $\L^+_{\exists}$ are preserved
whenever both $F$ and $F^*$ are unions. The interest of such a result
is that, in order to get an answer to query $\phi \in \L^+_{\exists}$
in the aggregated databases $F(\vec{D})$, we might run the query on
each instance $D$ separately, and then aggregate the results,
whichever is more efficient depending on the size of query $\phi$ and
instances $D_1, \ldots, D_n$.

Further, we may wonder whether a result symmetric to
 Lemma~\ref{existential} holds for intersections and the positive
 universal fragment $\L^+_{\forall}$ of first-order logic.
%
Unfortunately, in Example~\ref{ex1} we provided a formula $\phi
= \forall y P(x, y)$ in $\L^+_{\forall}$ and instances $D_1$, $D_2$
such that $ans(D_1 \cap D_2, \phi) \not \subseteq ans(D_1, \phi) \cap
ans(D_2, \phi)$. Hence, for $F$ and $F^*$ equal to set-theoretical
intersection, the diagram above does not commute for the query
language $\L^+_{\forall}$.

Nonetheless, we are able to prove a weaker but still significant
result related to Question~\ref{quest2}. Specifically, the next lemma
shows that
if in the diagram above $F$ is unanimous and the query language is
$\L^+_{\forall}$, then $F^*$ is unanimous, in the sense that
$\bigcap_{i \in \N} ans(D_i, \phi) \subseteq ans(\bigcap_{i \in \N}
D_i),\phi)$.
\begin{lemma} \label{universal}
Let aggregator $F$ be unanimous
and let $\L^+_{\forall}$ be the query language. Then, the induced
aggregator $F^*$ is also unanimous.
\end{lemma}
\begin{proof}
  We prove that $\bigcap_{i \in \N} ans(D_i, \phi) \subseteq
  ans(F(\vec{D}),\phi)$. So, if $\vec{u} \in \bigcap_{i \in \N}
  ans(D_i, \phi)$ then for every agent $i \in \N$,
  $(D_i, \vec{u}) \models \phi$.  We now prove by induction on
  $\phi \in \L^+_{\forall}$ that if for every $i \in \N$,
  $(D_i, \vec{u}) \models \phi$, then
  $(F(\vec{D}), \vec{u}) \models \phi$, that is $\vec{u} \in
  ans(F(\vec{D}),\phi)$.  

As to the base case for $\phi =
  P(x_1, \ldots, x_q)$ atomic, $(D_i, \vec{u}) \models \phi$ iff
  $\vec{u} \in D_i(P)$ for every $i \in \N$. In particular,
  $\vec{u} \in F(\vec{D})(P)$ as well by unanimity, and therefore
  $(F(\vec{D}), \vec{u}) \models P(x_1, \ldots, x_q)$.  The case for
  identity is immediate. 

As to the inductive case for $\phi
  = \psi \land \psi'$, suppose that for every $i \in \N$,
  $(D_i, \vec{u}) \models \phi$, that is,
  $(D_i, \vec{u}_1) \models \psi$ and $(D_i, \vec{u}_2) \models \psi$
  for suitable $ \vec{u}_1$ and $\vec{u}_2$ subsequences of
  $ \vec{u}$. By induction hypothesis we obtain that
  $(F(\vec{D}), \vec{u}_1) \models \psi$ and
  $(F(\vec{D}), \vec{u}_2) \models \psi'$, i.e.,
  $(F(\vec{D}), \vec{u}) \models \phi$.  Finally, if
  $(D_i, \vec{u}) \models \forall x \psi$ for every $i \in \N$, then
  for all $v \in \adom(D_i)$, $(D_i, \vec{u} \cdot
  v) \models \psi$. In particular, for all $v \in \adom(F(\vec{D}))$,
  $(D_i, \vec{u} \cdot v) \models \psi$ for every $i \in \N$, and by
  induction hypothesis, for all $v \in \adom(F(\vec{D}))$,
  $(F(\vec{D}), \vec{u} \cdot v) \models \psi$, i.e.,
  $(F(\vec{D}), \vec{u}) \models \forall x \psi$.
\end{proof}

Note that set-theoretical intersection is unanimous. 

A result symmetric to Lemma~\ref{universal} holds for language
 $\L^+_{\exists}$ and unions:
\begin{lemma} \label{existential2}
Let aggregator $F$ be grounded and let $\L^+_{\exists}$ be the query
language. Then, the induced aggregator $F^*$ is also grounded.
\end{lemma}
\begin{proof}
  We prove that $ans(F(\vec{D}), \phi)  \subseteq
  \bigcup_{i \in \N} ans(D_i, \phi)$. So, if
  $\vec{u} \in ans(F(\vec{D}),\phi)$ then 
$(F(\vec{D}), \vec{u}) \models \phi$.  We now prove by induction on
  $\phi \in \L^+_{\exists}$ that if
  $(F(\vec{D}), \vec{u}) \models \phi$, then for some $i \in \N$,
  $(D_i, \vec{u}) \models \phi$, and therefore
  $\vec{u} \in \bigcup_{i \in \N} ans(D_i, \phi)$.  

As to the base
  case for $\phi = P(x_1, \ldots, x_q)$ atomic,
  $(F(\vec{D}), \vec{u}) \models \phi$ iff $\vec{u} \in
  F(\vec{D})(P)$, iff $\vec{u} \in D_i(P)$ for some agents $i \in \N$
  by groundedness. In particular, $(D_i, \vec{u}) \models
  \phi$ as well.  The case for
  identity is immediate.
  
As to the inductive case for $\phi = \psi \lor \psi'$, suppose that
$(F(\vec{D}), \vec{u}) \models \phi$, i.e.,
$(F(\vec{D}), \vec{u}_1) \models \psi$ or
$(F(\vec{D}), \vec{u}_2) \models \psi'$ for suitable $ \vec{u}_1$ and
$\vec{u}_2$ subsequences of $ \vec{u}$. In the former case, by
induction hypothesis we have that
for some $i \in \N$, $(D_i, \vec{u}_1) \models \psi$, and
  therefore $(D_i, \vec{u}) \models \phi$.
  The case for
  $(F(\vec{D}), \vec{u}_2) \models \psi'$ is symmetric.
  Finally, if
  $(F(\vec{D}), \vec{u}) \models \exists x \psi$, then
  for some $v \in \adom(F(\vec{D}))$, $(F(\vec{D}), \vec{u} \cdot
  v) \models \psi$. In particular, 
by induction hypothesis,
$( D_i, \vec{u} \cdot v) \models \psi$ for
some $i \in \N$. 
Since $\psi$ is a positive formula,
$v \in \adom(D_i)$, and therefore, $(D_i, \vec{u} \cdot
v) \models \phi$.
\end{proof}

Note that Lemma~\ref{universal} and \ref{existential2} apply to all
quota rules, including union and intersection, though the query
language is rather limited.

We now move towards more practical query answering and consider the
language $\L_{CQ}$ of (unions of) conjuctive queries, which is a
popular query language in the theory of databases thanks to its
NP-complete query answering problem \cite{Chandra77}.
Formulas in $\L_{CQ}$ are defined according to the following BNF:
\begin{eqnarray*}
\phi  & ::= &  P_1(x_1, \ldots, x_{q_1}) \land \ldots \land P_m(x_1, \ldots, x_{q_m}) \mid 
\phi \lor \phi \mid \exists x \phi
\end{eqnarray*}

We now show that the result of conjunctive queries is preserved by
merge with incomplete information.
\begin{lemma} \label{lemma:merge}
Let aggregator $F$ be merge with incomplete information and let
$\L_{CQ}$ be the query language. Then, the induced aggregator $F^*$ is
also the merge rule.
\end{lemma}
\begin{proof}
  We show $F^*(ans(D_1, \phi), \ldots, ans(D_n, \phi))
  = \allowbreak ans(F(\vec{D}), \phi)$, where both $F$ and $F^*$ are
  the merge rule, by induction on $\phi \in \L_{CQ}$.
As to the base case for $\phi = P_1(x_1, \ldots,
  x_{q_1}) \land \ldots \land P_m(x_1, \ldots, x_{q_m})$,
  $(F(\vec{D}), \vec{u}) \models \phi$ iff $\vec{u}_1 \in
  F(\vec{D})(P_1)$, \ldots, $\vec{u}_m \in F(\vec{D})(P_m)$, where
  each $\vec{u}_i$ is a suitable subsequences of $\vec{u}$.  
  This is
  the case iff for every $j \leq n$, $\vec{u}'_{j,1} \in
  D_j(P_1)$, \ldots, $\vec{u}'_{j,m} \in D_j(P_m)$, where each
  $\vec{u}'_{j,i}$ differs from $\vec{u}_{i}$ as the latter might
  contain $\bot$ in designated positions. Again, the above is the case
  iff for every $j \leq n$, $(D_j, \vec{u}'_j) \models \phi$, that is,
  $\vec{u} \in F^*(\{\vec{u}'_{1}\}, \ldots, \{\vec{u}'_{n}\})$, where $F^*$ is the
  merge rule.
  For $\phi = \psi \lor \psi'$, $(F(\vec{D}), \vec{u}) \models \phi$
  iff $(F(\vec{D}), \vec{u}_1) \models \psi$ or
  $(F(\vec{D}), \vec{u}_2) \models \psi'$, where $\vec{u}_1$ and
  $\vec{u}_2$ are suitable subsequences of $\vec{u}$. By induction
  hypothesis, this is the case iff $\vec{u}_1 \in
  F^*(ans(D_1, \psi), \ldots, ans(D_n, \psi))$ or $\vec{u}_2 \in
  F^*(ans(D_1, \psi'), \ldots, ans(D_n, \psi'))$, that is, iff
  $\vec{u} \in F^*(ans(D_1, \phi), \allowbreak \ldots, ans(D_n, \phi))$.
Finally, for $\phi = \exists x \psi$, $\vec{u} \in ans(F(\vec{D}),\phi)$ iff
for some $v \in adom(F(\vec{D}))$, $(F(\vec{D}), \vec{u} \cdot
v) \models \psi$, iff $\vec{u} \cdot v \in
F^*(ans(D_1, \psi), \ldots, \allowbreak ans(D_n, \psi))$ by induction
hypothesis.  By definition of $F^*$, for every $j \leq n$,
$(D_j, \vec{u}' \cdot v') \models \psi$, where $\vec{u}' \cdot v'$
differs from $\vec{u} \cdot v$ as the latter might contain $\bot$ in
designated positions. The above is the case iff for every $j \leq n$,
$(D_j, \vec{u}') \models \phi$, iff $\vec{u} \in
F^*(ans(D_1, \phi), \ldots, ans(D_n, \phi))$, where $F^*$ is the merge
rule.
\end{proof}

By Lemma~\ref{lemma:merge} we can query the individual instances and then
merge the corresponding answers instead of querying the merged
database.
%

%

%
By using the relation-wise average voter rule, we are able to prove
the following preservation result.  Hereafter, the average $Ave$ of
answers $ans(D_1, \phi), \ldots, ans(D_n, \phi)$ is computed as follows:
{\small
\begin{eqnarray} 
Ave(ans(D_1, \phi), \ldots, ans(D_n, \phi)) \!  \!  \!   \!   \!  & =
\!  \!  \!   \!   \! \!  \!  \!   \!   \! & \argmin_{ans(D_i, \phi) \mid i\in\N\}} \sum_{j \in \N}
(|ans(D_j, \phi) \setminus ans(D_i, \phi)| + |ans(D_i, \phi) \setminus
ans(D_j, \phi)|) \ \ \ \ \ \   \label{ag11}
\end{eqnarray}
}

Note that the relation-wise average voter rule can be defined as the union of the
averages of the individual relations associated to each $P\in \mathcal D$. 

%
\begin{lemma} \label{average}
Let aggregator $F$ be the average rule and let first-order logic $\L$
be the query language. Then, $F^*(ans(D_1, \phi), \ldots,
ans(D_n, \phi)) = ans(F(\vec{D}), \phi)$ is a subset of 
$Ave(ans(D_1, \phi), \ldots, \allowbreak ans(D_n, \phi))$.
\end{lemma}
\begin{proof}
If $\vec{u} \in F^*(ans(D_1, \phi), \ldots, ans(D_n, \phi)) =
  ans(F(\vec{D}), \phi)$ then $(F(\vec{D}), \vec{u}) \models \phi$.
  Now we prove on induction of the structure of $\phi$ that if
  $(F(\vec{D}), \vec{u}) \models \phi$ then $\vec{u}$ belongs to the average of
  $ans(D_1, \phi), \ldots, \allowbreak ans(D_n, \phi)$.
For $\phi = P(x_1, \ldots, x_{q})$, if
$(F(\vec{D}), \vec{u}) \models \phi$ then $\vec{u} \in
  F(\vec{D})(P)$, and therefore $\vec{u}$ belongs to the
average of $ans(D_1, \phi), \ldots, ans(D_n, \phi)$, as $F(\vec{D})$
minimises the distance for all $P \in \D$.
  For $\phi = \psi \star \psi'$, where $\star$ is a Boolean operator,
  $(F(\vec{D}), \vec{u}) \models \phi$ iff
  $(F(\vec{D}), \vec{u}_1) \models \psi \ \hat{\star} \
  (F(\vec{D}), \vec{u}_2) \models \psi'$, where $\hat{\star}$ is the
  interpretation of $\star$ and $\vec{u}_1$, $\vec{u}_2$ are suitable
  subsequences of $\vec{u}$. By induction hypothesis,
then $\vec{u}_1$ and $\vec{u}_2$ minimise the distances in the answers
to queries $\psi$ and $\psi'$ respectively, then $\vec{u}$ does so for
$\phi$, and therefore $\vec{u}$ belongs to the average of
$ans(D_1, \phi), \ldots, ans(D_n, \phi)$.
%
Finally, universal (resp.~existential) quantification is dealt with by
considering it as a finite conjunction (resp.~disjunction).
\end{proof}

To conclude this section we discuss the results obtain so far, which
can be seen as a first contribution on the relationship between
database aggragation and query answering.  In particular,
Lemma~\ref{existential} can be seen as a (partial) answer to
Question~\ref{quest1}. Similarly, Lemma~\ref{universal}
and \ref{existential2} are related to Question~\ref{quest2}. However,
all the applicability of these results is restricted by the limited
expressivity of the query languages.  On the other hand,
Lemma~\ref{lemma:merge} and \ref{average} show that merge with
incomplete information and the average voter rule preserve (union of)
conjuctive queries and the whole of first-order logic respectively.
Results along these lines
may find application in efficient query answering: it might be that in
selected cases, rather than querying the aggregated database
$F(\vec{D})$, it is more efficient to query the individual instances
$D_1, \ldots, D_n$ and then aggregate the answers. Then, it is crucial
to know which answers are preserved by the different aggregation
procedures. The results provided in this section aimed to be a first,
preliminary step in this direction.

%

%% file: conc.tex
In this paper we have proposed a framework for the aggregation of
conflicting information coming from multiple sources in the form of
finite relational databases.  We proposed a number of aggregators
inspired by the literature on social choice theory, and adapted a
number of axiomatic properties.  We then focused on two natural
questions which arise when dealing with the aggregation of
databases. First, in Section~\ref{sec:collectiverationality} we studied
what kind of integrity constraints are lifted by some of the rules we
proposed, i.e., what constraints are true in the aggregated database
supposing that all individual input satisfies the same constraints.
Second, in Section~\ref{sec:queries} we investigated first-order query
answering in the aggregated databases, characterising some languages
for which the aggregation of the answers in the individual databases
correspond to the answer to the query on the aggregated database.

Our initial results shed light on the possible use of choice-theoretic
techniques in database merging and integration, and opens multiple
interesting directions for future research.  In particular, the
connections to the literature on aggregation and merging can be
investigated further. Firstly,
Section~\ref{sec:collectiverationality} showcased results for which
database aggregation behaves similarly to binary aggregation with
integrity constraints (see \cite{GrandiEndrissAIJ2013}), but pointed
at some crucial differences.  In particular, there are natural
classes of integrity constraints used in databases for which the
equivalent in propositional logic, the language of choice for binary
aggregation, would be tedious and lenghty. We were able to provide
initial results on their preservation through aggregation.  Secondly,
the recent work of \cite{EndrissGrandiAIJ2017} is also strongly
related to our contribution. Since graphs are a specific type of
relational structures, our work directly generalise their graph
aggregation framework to relations of arbitrary arity. However, the
specificity of their setting allows them to obtain very powerful
impossibility results, which are yet to be explored in the area of
database aggregation.  Thirdly, to the best of our knowledge the
problem of aggregated query answering is new in the literature on
aggregation, albeit a similar problem has been studied in the
aggregation of argumentation graphs \cite{ChenEndrissTARK2017}.
Also this direction deserves further investigation.